\documentclass{amsart}

\usepackage[T1]{fontenc}
\usepackage[margin=1in]{geometry}
\usepackage{color}

\usepackage{amsmath,amssymb,amsthm}
\usepackage{amstext}
\usepackage{tikz}
\usepackage{comment}
\usepackage{mathtools}
\usepackage{mdframed}
\usepackage[lofdepth,lotdepth,labelformat=simple]{subfig}

\usepackage[comma,square,numbers,sort]{natbib}
\usepackage[colorlinks]{hyperref}
\usepackage[all]{hypcap}
\usepackage{enumitem}
\usepackage[foot]{amsaddr}
\mdfsetup{nobreak=true,skipabove=8pt,skipbelow=8pt,innertopmargin=8pt,innerbottommargin=8pt,frametitleaboveskip=8pt,frametitlebelowskip=0pt}

\DeclareCaptionListOfFormat{myformat}{\textsc{#2}}
\newcommand{\eq}[1]{\hyperref[eq:#1]{(\ref*{eq:#1})}}
\renewcommand{\sec}[1]{\hyperref[sec:#1]{Section~\ref*{sec:#1}}}
\newcommand{\app}[1]{\hyperref[app:#1]{Appendix~\ref*{app:#1}}}
\newcommand{\thm}[1]{\hyperref[thm:#1]{Theorem~\ref*{thm:#1}}}
\newcommand{\lem}[1]{\hyperref[lem:#1]{Lemma~\ref*{lem:#1}}}
\newcommand{\cor}[1]{\hyperref[cor:#1]{Corollary~\ref*{cor:#1}}}
\newcommand{\defn}[1]{\hyperref[defn:#1]{Definition~\ref*{defn:#1}}}
\newcommand{\ex}[1]{\hyperref[ex:#1]{Example~\ref*{ex:#1}}}
\newcommand{\fct}[1]{\hyperref[fct:#1]{Fact~\ref*{fct:#1}}}
\newcommand{\fig}[1]{\hyperref[fig:#1]{Figure~\ref*{fig:#1}}}
\newcommand{\figs}[1]{\hyperref[fig:#1]{Figures~\ref*{fig:#1}}}

\theoremstyle{definition}
\newtheorem{definition}{Definition}
\newtheorem{problem}{Problem}
\theoremstyle{plain}
\newtheorem{theorem}{Theorem}
\newtheorem{lemma}{Lemma}
\newtheorem{corollary}{Corollary}
\newtheorem{fact}{Fact}

\newcommand{\microspace}{\mspace{0.5mu}}

\newcommand{\sket}[1]{|\microspace #1 \microspace \rangle}
\newcommand{\sbra}[1]{\langle\microspace #1 \microspace |}

\newcommand{\sden}[1]{\sket{#1}\mspace{-4mu}\sbra{#1}}

\newcommand{\ket}[1]{\left|\microspace #1 \microspace \right\rangle}
\newcommand{\bra}[1]{\left\langle\microspace #1 \microspace \right|}

\newcommand{\den}[1]{\ket{#1}\mspace{-4mu}\bra{#1}}

\newcommand{\II}{\mathbb{I}}

\newcommand{\Span}{\operatorname{span}}

\title{Complexity of the XY antiferromagnet at fixed magnetization}
\author{Andrew M. Childs$^{1,2,3}$} \email{amchilds@umd.edu}
\author{David Gosset$^{1,2,4}$} \email{dngosset@gmail.com}
\author{Zak Webb$^{2,5}$} \email{zakwwebb@gmail.com}
\address{$^1$ Department of Combinatorics \& Optimization, University of Waterloo}
\address{$^2$ Institute for Quantum Computing, University of Waterloo}
\address{$^3$ Department of Computer Science,
Institute for Advanced Computer Studies, and
Joint Center for Quantum Information and Computer Science,
University of Maryland}
\address{$^4$ Walter Burke Institute for Theoretical Physics and Institute for Quantum Information and Matter, California Institute of Technology}
\date{}
\address{$^5$ Department of Physics \& Astronomy, University of Waterloo}
\date{}

\begin{document}

\maketitle

\begin{abstract}
We prove that approximating the ground energy of the antiferromagnetic XY model on a simple graph at fixed magnetization (given as part of the instance specification) is QMA-complete. To show this, we strengthen a previous result by establishing QMA-completeness for approximating the ground energy of the Bose-Hubbard model on simple graphs. Using a connection between the XY and Bose-Hubbard models that we exploited in previous work, this establishes QMA-completeness of the XY model.
\end{abstract}

\section{Introduction}

Kitaev pioneered the study of quantum constraint satisfaction problems, where the goal is to approximate the ground energy of a Hamiltonian \cite{KSV02}. Many examples of such ground energy problems are known to be complete for the complexity class QMA, a quantum analogue of NP (see, e.g., \cite{Boo12}).

In this paper, we focus on computational problems defined by graphs.  For example, MAX-CUT is a classical constraint satisfaction problem defined by a graph. The goal is to find a subset of vertices that maximizes the number of edges between that subset and its complement. This can be rephrased as a ground energy problem: it is equivalent to minimizing the energy of the Ising antiferromagnet on the graph, where there is a bit for every vertex and a constraint penalizing adjacent bits that agree. Equivalently, we may consider a qubit at every vertex and a $ZZ$ interaction for every edge.

To obtain a genuinely quantum constraint satisfaction problem defined by a graph, we consider interaction terms with nonzero off-diagonal matrix elements.  Natural choices include the antiferromagnetic Heisenberg model, with an $XX+YY+ZZ$ interaction for each edge, and the antiferromagnetic XY model, with an $XX+YY$ interaction for each edge. The complexities of the ground energy problems for these models are unknown, although some variants have been studied \cite{SV09, CM13, BHQMA}. Recent work has established QMA-completeness for the antiferromagnetic XY model with coefficients that vary throughout the graph and depend on system size \cite{PM15}. This result holds even when the graph is a triangular lattice in two dimensions. Our result is incomparable because while we consider general (simple) graphs, we restrict the coefficient for every edge to be $1$. In addition, the restriction to fixed magnetization is not present in reference \cite{PM15}.

Our main result concerns the antiferromagnetic XY model on a graph.  For a given simple graph $G$ with vertex set $V(G)$ and edge set $E(G)$, the Hamiltonian has the form
\begin{equation}
  \frac{1}{2} \sum_{\{i,j\}\in E(G)} (X_i X_j + Y_i Y_j)
\label{eq:ham_intro}
\end{equation}
where $X,Y,Z$ denote Pauli matrices and a subscript indicates which qubit is acted on.
Note that this Hamiltonian commutes with the total magnetization operator $M_z = \sum_{i \in V(G)} Z_i$, so it decomposes into sectors for each eigenvalue of $M_z$. We prove that approximating the ground energy of this Hamiltonian in a sector with fixed magnetization is a QMA-complete problem.

Our result is a natural extension of our previous work \cite{BHQMA}. The difference is that we previously considered Hamiltonians defined by graphs that may have self-loops. In that context we established QMA-completeness of the ground energy problem for the Bose-Hubbard model (at fixed particle number), a system of bosons hopping on a graph with a repulsive on-site interaction. Then, using the relationship between hard-core bosons and spins, we also established QMA-completeness of a ground energy problem related to the XY model, where the Hamiltonian has an $XX+YY$ term associated with each edge (as in \eq{ham_intro}) as well as a local magnetic field  associated with each self-loop. In this paper we present a stronger result since we consider only simple graphs and Hamiltonians of the form  \eq{ham_intro}.

Our proof relies heavily on machinery and results from reference \cite{BHQMA}. We prove QMA-hardness of the ground energy problem for the Bose-Hubbard model on a \emph{simple} graph at fixed particle number. Our starting point is the previous QMA-completeness result \cite{BHQMA} which pertains to graphs $G$ with self-loops. Counterintuitively, we begin by increasing the number of self-loops: we modify a graph $G$ from the previous construction to obtain a new graph $G^{\rm SL}$ in which every vertex has a self-loop. The new graph $G^{\rm SL}$ is formed by taking two copies of $G$, adding self-loops to each of them, and then adding some edges between the two copies. The modification is performed in such a way that the ground spaces of the Bose-Hubbard model on $G$ and $G^{\rm SL}$ (in the sector with a given number of particles) are simply related.  Once we have a self-loop at every vertex, we then remove all self-loops to obtain another graph $G^{\rm NSL}$ with no self-loops, which was our goal. This removal is equivalent to subtracting a term in the Hamiltonian that is proportional to the identity.  We thus obtain a graph $G^{\rm NSL}$ with no self-loops such that approximating the ground energy of the Bose-Hubbard model on $G^{\rm NSL}$ is as hard as approximating the ground energy of $G$, which is QMA-hard by our previous result.

Finally, we use a reduction presented in \cite{BHQMA} showing that an instance of this ground energy problem {for the Bose-Hubbard model} on a graph is equivalent to an instance of the ground energy problem for the XY model at fixed magnetization on the same graph. Since this reduction preserves the graph, we establish QMA-completeness for the XY model as discussed above.

The remainder of this paper is organized as follows.  In \sec{prelim}, we provide basic definitions and tools used in this paper. In \sec{previous}, we review results from reference \cite{BHQMA}. In \sec{adding_self_loops}, we describe a procedure that modifies graphs from the previous QMA-completeness result so that every vertex has a self-loop. We show a relationship between the ground energies before and after the modification. Finally, we remove all the self-loops from the modified graph, giving only a constant overall energy shift in the associated Hamiltonian. In \sec{removing_self_loops}, we use this strategy to establish that the ground energy problem for the antiferromagnetic XY model (on simple graphs, at fixed magnetization) is QMA-complete.

\section{Preliminaries}
\label{sec:prelim}

In this paper, $G$ denotes a graph with vertex set $V(G)$, edge set $E(G)$, and adjacency matrix $A(G)$.  Later we will be interested in the case where $G$ is a simple graph, but for now we allow the possibility that it has at most one self-loop per vertex. In other words, $A(G)$ can be any symmetric 0-1 matrix. 

\subsection{The antiferromagnetic XY model on a graph}

We define the antiferromagnetic XY model with local magnetic fields on $G$ to be the $|V(G)|$-qubit, two-local Hamiltonian
\begin{align}
  O_G &= \sum_{A(G)_{ij}=1, i\neq j} (\sket{01}\sbra{10} + \sket{10}\sbra{01})_{ij} + \sum_{A(G)_{ii}=1} \sket{1}\sbra{1}_i\\
    &=  \sum_{A(G)_{ij}=1,i\neq j} \frac{X_i X_j + Y_i Y_j}{2} + \sum_{A(G)_{ii} = 1} \frac{1 - Z_i}{2}.
\label{eq:OG}
\end{align}
Note that the second term is only present if the graph $G$ has self loops; this term vanishes for simple graphs, giving the usual antiferromagnetic XY model.  It is easy to see that this Hamiltonian (either with or without the second, local magnetic field term) conserves Hamming weight.  Let
\[
\mathrm{Wt}_N=\Span\{|z\rangle: z\in \{0,1\}^{|V(G)|}, \, \mathrm{wt}(z)=N\}
\]
be the subspace with Hamming weight $N$. We write $\theta_N(G)$ for the smallest eigenvalue of $O_G$ within the sector with Hamming weight $N$ (i.e., the smallest eigenvalue of the restriction $O_G|_{\mathrm{Wt}_N}$).
  
\subsection{The Bose-Hubbard model on a graph}

We now review the Hamiltonian of the Bose-Hubbard model on a graph, as defined in \cite{BHQMA}. We present only the ``first quantized'' formulation of this model that we use in this paper; see \cite{BHQMA} for a broader discussion.

The Hilbert space of $N$ distinguishable particles that live on the vertices of $G$ is 
\[
\bigl(\mathbb{C}^{|V(G)|}\bigr)^{\otimes N}=\Span\{ |i_{1}i_{2}\ldots i_{N}\rangle:\, i_{j}\in V(G)\}.
\]
Here each register represents the location of a particle. The Hilbert space of $N$ indistinguishable bosons on $G$
is the subspace of symmetric states
\[
\mathcal{Z}_{N}(G)=\Span\{ \text{Sym}(|i_{1}i_{2}\ldots i_{N}\rangle) : i_{j}\in V(G)\} 
\]
 where 
\[
\text{Sym}(|i_{1}i_{2}\ldots i_{N}\rangle)=\frac{1}{\sqrt{N!}}\sum_{\pi\in S_{N}}|i_{\pi(1)}i_{\pi(2)}\ldots i_{\pi(N)}\rangle
\]
and $S_{N}$ is the symmetric group on $N$ elements. 

For any operator $M\in \mathbb{C}^{|V(G)|}$, write
\[
M^{(i)}=\underbrace{1\otimes \cdots \otimes 1}_{i-1} \otimes \, M \otimes \underbrace{1 \otimes \cdots \otimes 1}_{N-i}.
\]
for the operator that acts on $(\mathbb{C}^{|V(G)|})^{\otimes N}$ as $M$ on the $i$th register and as the identity on all other registers. Define 
\begin{equation}
H_{G}^{N}=\sum_{i=1}^{N} A(G)^{(i)} +\sum_{k\in V(G)} \hat{n}_k(\hat{n}_k-1)
\label{eq:HGN}
\end{equation}
where
\[
\hat{n}_k=\sum_{i=1}^N |k\rangle\langle k|^{(i)}
\]
is an operator that counts the number of particles at vertex $k$. The Hamiltonian \eq{HGN} acts on the distinguishable-particle Hilbert space. Since it is symmetric under permutation of the $N$ registers, the bosonic space $\mathcal{Z}_N(G)$ is an invariant subspace for $H_G^{N}$. The $N$-particle Bose-Hubbard model is the restriction of this Hamiltonian to the bosonic subspace
\[
\bar{H}_{G}^{N}=H_G^{N}\big|_{\mathcal{Z}_N(G)}.
\]
Looking at equation \eq{HGN}, we see that the first term (the movement, or hopping, term) has smallest eigenvalue equal to $N$ times the smallest eigenvalue of $A(G)$, which we denote $\mu(G)$, while the second term (the interaction term) is positive semidefinite. Hence the smallest eigenvalue of $\bar{H}_{G}^{N}$ is at least that of $H_G^{N}$, which is at least $N\mu(G)$.

It is convenient to subtract this constant to make the Hamiltonian positive semidefinite. Define
\[
H(G,N)=\bar{H}_{G}^{N}-N\mu(G),
\]
and write $\lambda_N^1(G)\geq 0$ for its smallest eigenvalue. When $\lambda_N^1(G)=0$ the $N$-particle ground space minimizes the energy of both terms (movement and interaction) separately, and we say the Hamiltonian is \emph{frustration free}.

There is a connection between the Bose-Hubbard model on a graph and the XY model on the same graph.  Consider the subspace of bosonic $N$-particle states that have zero energy for the second (interaction) term in \eq{HGN}. (For example, any frustration-free state lives in this subspace.) States in this subspace have no support on basis states where more than one particle occupies any vertex of the graph. This is the subspace of \emph{hard-core bosons}.  Since every vertex can be occupied by at most one particle, this subspace can be identified with the Hamming weight $N$ subspace of $|V(G)|$ qubits (each qubit represents a vertex, with basis states $|0\rangle, |1\rangle$ representing unoccupied and occupied states, respectively). Thus the action of the Bose-Hubbard model in the subspace of hard-core bosons is equivalent to a spin model. In fact, as discussed in \cite{BHQMA}, the restriction of \eq{HGN} to the subspace of hard-core bosons is exactly equal to the restriction of $O_G$ (given by \eq{OG}) to the Hamming weight $N$ space $\mathrm{Wt}_N$.

\subsection{Gate graphs}

The QMA-completeness construction from reference \cite{BHQMA} uses a class
of graphs called gate graphs that we now review. 

\subsubsection{The graph $g_{0}$}

A $128$-vertex simple graph denoted $g_{0}$ plays a central role in the construction. The graph $g_{0}$ is defined
explicitly in reference \cite{BHQMA} by specifying its adjacency matrix $A(g_{0})$; here we only describe the properties that we use in this paper. The vertices of $g_{0}$
are labeled by tuples
\[
(z,t,j) : z\in\{0,1\},\, t\in[8],\, j\in\{0,\ldots,7\}
\]
where $[n] = \{1,\ldots,n\}$.
The adjacency matrix $A(g_{0})$ acts on the Hilbert space 
\[
\Span\{ |z\rangle|t\rangle|j\rangle: z\in\{0,1\},\, t\in[8],\, j\in\{0,\ldots,7\}\} .
\]
The smallest eigenvalue of $A(g_{0})$ is $e_{1}=-1-3\sqrt{2}$.  The corresponding
eigenspace has an orthonormal basis given by the four states
\begin{align}
|\psi_{z,0}\rangle & =\frac{1}{\sqrt{8}}\bigl(|z\rangle(|1\rangle+|3\rangle+|5\rangle+|7\rangle)+H|z\rangle(|2\rangle+|8\rangle)+HT|z\rangle(|4\rangle+|6\rangle)\bigr)|\omega\rangle\label{eq:psi1}\\
|\psi_{z,1}\rangle & =|\psi_{z,0}\rangle^{*}
\label{eq:psi2}
\end{align}
where $z\in \{0,1\}$, 
\[
H=\frac{1}{\sqrt{2}}\begin{pmatrix}
1 & 1\\
1 & -1
\end{pmatrix} \qquad T=\begin{pmatrix}
1 & 0\\
0 & e^{i\frac{\pi}{4}}
\end{pmatrix}
\]
 and 
\[
|\omega\rangle=\frac{1}{\sqrt{8}}\sum_{j=0}^{7}e^{-i\frac{\pi}{4}j}|j\rangle.
\]

Note that, depending on the value of $t$ in the second register, the first register of the state $|\psi_{z,0}\rangle$ contains the output of a single-qubit computation where either the identity, Hadamard, or $HT$ gate is applied to the state $|z\rangle$. (The state of the third register is in a product state with the first two, and is somewhat uninteresting; this register exists for technical reasons.)

\subsubsection{Gate graphs and gate diagrams}

A gate diagram is a schematic representation of a gate graph.
To define gate graphs, we first define gate diagrams
and then describe how a graph is associated with each of them. 

The simplest gate diagrams are shown in \fig{diagram_elements}. These three
basic gate diagrams are also called diagram elements. The diagram
elements each represent the same gate graph which is just
the $128$-vertex graph $g_{0}$ described above. Each diagram element
has a unitary label which is either $1$, $H$, or $HT$, as well
as a set of eight circles that we call nodes. A node labeled
$(z,t)$ is associated with the eight vertices of $g_{0}$ labeled $(z,t,j)$
with $j\in\{0,\ldots,7\}$. Note that only half of the $16$ possible nodes $(z,t)$ appear in a given diagram element. Moreover, the half that does appear depends on the unitary label $1$, $H$, or $HT$.

A gate diagram is constructed by taking a set of diagram elements
and adding edges between some pairs of nodes and self-loops to other
nodes. Each node may have a self loop or an incident edge but never
both (and never more than one edge or self-loop). If the gate diagram
has $R$ diagram elements, then each node can be labeled $(q,z,t)$
with $q\in[R]$, $z\in\{0,1\}$ and $t\in[8]$. We write $\mathcal{S}$
for the set of nodes $(q,z,t)$ that have self-loops attached and
we write $\mathcal{E}$ for the set of pairs of nodes $\{(q,z,t),(q^{\prime},z^{\prime},t^{\prime})\}$
that are connected by an edge.

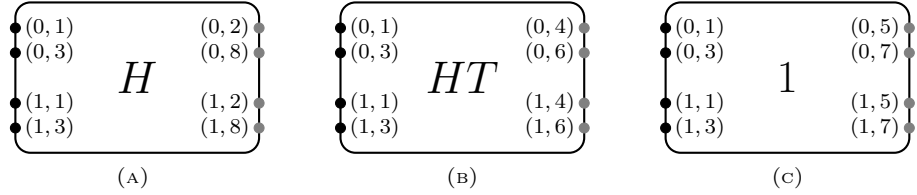
\begin{figure}
\centering
\subfloat[][]{ 
\label{fig:diagram_elementH}
\begin{tikzpicture}
  \draw[rounded corners=2mm,thick] (0,0) rectangle (3.24 cm, 2cm);

 \foreach \y in {.33,.66,1.33,1.66}
{ 
 \foreach \x/\color in {0/black,3.24/gray}
{      
\draw[fill=\color,draw=\color] (\x cm, \y cm) circle (.66mm);   
}
}   

\foreach \z in {0,1}
{    
\node[right] at (0,1.66-\z) {\footnotesize $(\z,1)$}; 
\node[right] at (0,1.33-\z) {\footnotesize $(\z,3)$};    

   \node[left] at (3.24,1.66-\z) {\footnotesize $(\z,2)$}; 
   \node[left] at (3.24,1.33-\z) {\footnotesize $(\z,8)$};  
}   

  \node at (1.62,1) {\huge $H$};   
\end{tikzpicture} 
} 
\qquad
\subfloat[][]{
\label{fig:diagram_elementHT}
\begin{tikzpicture}
  \draw[rounded corners=2mm,thick] (0,0) rectangle (3.24 cm, 2cm);
   
\foreach \y in {.33,.66,1.33,1.66}
{  
\foreach \x/\color in {0/black,3.24/gray}
{     
\draw[fill=\color,draw=\color] (\x cm, \y cm) circle (.66mm);   
}
}     

\foreach \z in {0,1}
{    
\node[right] at (0,1.66-\z) {\footnotesize $(\z,1)$}; 
   \node[right] at (0,1.33-\z) {\footnotesize $(\z,3)$};     
  
  \node[left] at (3.24,1.66-\z) {\footnotesize $(\z,4)$}; 
   \node[left] at (3.24,1.33-\z) {\footnotesize $(\z,6)$}; 
 }    
  
\node at (1.62,1) {\huge $HT$}; 
 \end{tikzpicture}
} 
\qquad 
\subfloat[][]
{
\label{fig:diagram_element1}
\begin{tikzpicture}
  \draw[rounded corners=2mm,thick] (0,0) rectangle (3.24 cm, 2cm);

\foreach \y in {.33,.66,1.33,1.66}
{   
\foreach \x/\color in {0/black,3.24/gray}
{     
\draw[fill=\color,draw=\color] (\x cm, \y cm) circle (.66mm);   
}
}    

\foreach \z in {0,1}
{     
\node[right] at (0,1.66-\z) {\footnotesize $(\z,1)$};   
 \node[right] at (0,1.33-\z) {\footnotesize $(\z,3)$};    

 \node[left] at (3.24,1.66-\z) {\footnotesize $(\z,5)$};  
  \node[left] at (3.24,1.33-\z) {\footnotesize $(\z,7)$};   }    
  
\node at (1.62,1) {\huge $1$};  
\end{tikzpicture}  
}

\caption{Diagram elements. \label{fig:diagram_elements}}
\end{figure}

A gate diagram with $R$ diagram elements, self-loop set $\mathcal{S}$,
and edge set $\mathcal{E}$ is associated with a gate graph $G$
as follows. The vertex set of $G$ is just $R$ copies of the vertex
set of $g_{0}$, with vertices labeled $(q,z,t,j)$ with $q\in[R]$, $z\in\{0,1\}$, $t\in[8]$, and $j\in\{0,\ldots,7\}$.
The adjacency matrix $A(G)$ acts on the Hilbert space 
\[
\Span\{ |q\rangle|z\rangle|t\rangle|j\rangle: q\in[R],\: z\in\{0,1\},\, t\in[8],\, j\in\{0,\ldots,7\}\} 
\]
and is defined by 
\begin{align}
A(G) & =1_{q}\otimes A(g_{0})+h_{\mathcal{S}}+h_{\mathcal{E}}\label{eq:AG}\\
h_{\mathcal{S}} & =\sum_{(q,z,t)\in \mathcal{S}}|q,z,t\rangle\langle q,z,t|\otimes1_{j}\label{eq:HS}\\
h_{\mathcal{E}} & =\sum_{\{(q,z,t),(q',z',t')\}\in \mathcal{\mathcal{E}}}(|q,z,t\rangle+|q^{\prime},z^{\prime},t^{\prime}\rangle)(\langle q,z,t|+\langle q^{\prime},z^{\prime},t^{\prime}|)\otimes1_{j}\label{eq:HE}
\end{align}
where we write $1_{q}$ and $1_{j}$ for the identity operators on
the first and fourth registers, respectively. Note that $h_{\mathcal{S}}$
and $h_{\mathcal{E}}$ are always positive semidefinite, so the smallest
eigenvalue of $A(G)$ is lower bounded by $e_{1}$, the smallest eigenvalue
of $A(g_{0})$. When these quantities are equal, we say $G$ is an
$e_{1}$-gate graph.

\begin{definition}
A gate graph $G$ is an $e_{1}$-gate graph if its smallest
eigenvalue is $e_{1}=-1-3\sqrt{2}$.
\end{definition}

Suppose $\ket{\chi}$ is a ground state of the adjacency matrix of an $e_1$-gate graph $G$. For each $q\in [R]$, define
\[
|\psi^q_{z,a}\rangle=|q\rangle|\psi_{z,a}\rangle
\]
(where the states $|\psi_{z,a}\rangle$ are defined in \eq{psi1}--\eq{psi2}) and note (from equation \eq{AG}) that $\ket{\chi}\in \mathcal{Y}$, where
\begin{equation}
\mathcal{Y}=\Span\{|\psi^q_{z,a}\rangle: q\in [R],\, z,a\in \{0,1\}\}
\label{eq:Yspace}
\end{equation}
and $h_\mathcal{S}\ket{\chi}=h_\mathcal{E}\ket{\chi}=0$.

\subsection{Spectral bounds for positive semidefinite matrices}  

Throughout our proof, we bound eigenvalue gaps of positive semidefinite matrices.  
For a positive semidefinite matrix $A$, let $\gamma(A)$ denote the smallest nonzero eigenvalue.

\begin{fact}
Let $H_A$ and $H_B$ be positive semidefinite matrices.  Let $H_A$ have nonempty nullspace $S$.  Then $\gamma(H_A+H_B) \leq \gamma(H_B|_{S})$.
  \label{fct:variation_gap}
\end{fact}
\begin{proof}
Let $\ket{\eta}$ be an eigenvector of $H_B|_{S}$ with eigenvalue $\gamma(H_B|_{S})$. 

Suppose first that the nullspace of $H_A+H_B$ is nonempty. In this case the nullspace of $H_A+H_B$ is equal to the nullspace of $H_B|_{S}$ (since both $H_A$ and $H_B$ are positive semidefinite), and hence $\ket{\eta}$ is orthogonal to this space. Thus
\[
\gamma(H_A+H_B)\leq \bra{\eta}H_A+H_B\ket{\eta}=\bra{\eta}H_B\ket{\eta}=\gamma(H_B|_{S}).
\]

If instead the nullspace of $H_A+H_B$ is empty, then $\gamma(H_A+H_B)$ is the smallest eigenvalue of $H_A+H_B$ and again is variationally upper bounded by $\bra{\eta}H_A+H_B\ket{\eta}=\bra{\eta}H_B\ket{\eta}=\gamma(H_B|_{S})$.
\end{proof}

A version of the following lemma was used but not explicitly stated in reference \cite{MLM99}. We gave a proof of this ``Nullspace Projection Lemma'' in \cite{BHQMA} and used it extensively in that work. We recently became aware of another work \cite{AFH09} that proves a slightly stronger bound than the one from \cite{BHQMA}. We quote the better bound here.

\begin{lemma}[\textbf{Nullspace Projection Lemma }\cite{AFH09}]
\label{lem:npl}Let $H_A, H_B\geq0$. Suppose the nullspace $S$ of $H_A$ is nonempty and 
\[
\gamma(H_{B}|_{S})\geq c > 0\qquad\text{and}\qquad\gamma(H_{A})\geq d >0.
\]
Then 
\begin{equation}
\gamma(H_{A}+H_{B})\geq\frac{cd}{d+\|{H_{B}}\|}.\label{eq:npl_lower_bnd}
\end{equation}
\end{lemma}

\section{Previous results}\label{sec:previous}

In this Section we summarize results from reference \cite{BHQMA} that are used in this paper.

In \cite{BHQMA} we considered the problem of approximating the ground energy of the Bose-Hubbard model on a graph at fixed particle number. We also considered a special case of this problem where the goal is to determine whether the Hamiltonian is close to being \emph{frustration free} (recall that in our setting, frustration freeness means $\lambda_N^1(G)=0$).

\begin{mdframed}
\begin{problem}
[\textbf{$\alpha$-Frustration-Free Bose-Hubbard Hamiltonian}]
We are given a $K$-vertex graph $G$, a number of particles $N\leq K$, and a precision parameter $\epsilon=\frac{1}{T}$. The integer $T \ge 4K$ is provided in unary; the graph is specified by its adjacency matrix, which can be any $K\times K$ symmetric $0$-$1$ matrix. We are promised that either $\lambda_{N}^{1}(G)\leq\epsilon^\alpha$ (yes instance) or $\lambda_{N}^{1}(G)\geq \epsilon+\epsilon^\alpha$ (no instance) and we are asked to decide which is the case. 
\end{problem}
\end{mdframed}

Here $T$ is given in unary so that the input size scales linearly with $T$.  An algorithm for this problem is considered efficient if it uses resources polynomial in $\frac{1}{\epsilon}=T$.

The positive integer $\alpha$ in the above definition parameterizes how close to frustration free the Hamiltonian is in the yes case. This definition slightly generalizes one presented in \cite{BHQMA}; the definition of the ``Frustration-Free Bose-Hubbard Hamiltonian'' problem given in that paper corresponds to the choice $\alpha=3$. Here it is convenient to explicitly define computational problems for each positive integer $\alpha$. The proof of QMA-completeness presented in reference \cite{BHQMA} applies to each $\alpha$ (as noted on page 11 of \cite{BHQMA}).

\begin{theorem}[\cite{BHQMA}]
For any positive integer $\alpha$, the problem $\alpha$-Frustration-Free Bose-Hubbard Hamiltonian is QMA-complete.
\label{thm:BHQMA_bounds}
\end{theorem} 

In fact, the results of reference \cite{BHQMA} can be used to show that each of these problems is QMA-hard when restricted to a certain subset of instances. In particular,

\begin{corollary}
\label{cor:specialcase}
For any positive integer $\alpha$, the problem $\alpha$-Frustration-free Bose-Hubbard Hamiltonian remains QMA-hard if we additionally promise that the graph $G$ is an $e_1$-gate graph described by a gate diagram with $R$ diagram elements,  satisfying 
\begin{equation}
\gamma(A(G)-e_1)\geq \frac{C_0}{R^3}
\label{eq:condition}
\end{equation}
where $C_0$ is an absolute constant.
\label{cor:gap}
\end{corollary} 
In \app{cor} we prove this Corollary using the results of \cite{BHQMA}.

Using a reduction based on the connection between spins and hard-core bosons, we established QMA-completeness of the following problem.

\begin{mdframed}
\begin{problem}
[\textbf{XY Hamiltonian with local magnetic fields}]
We are given a $K$-vertex graph $G$, an integer $N\leq K$, a real number $c$, and a precision parameter $\epsilon=\frac{1}{T}$. The positive integer $T$ is provided in unary; the graph is specified by its adjacency matrix, which can be any $K\times K$ symmetric $0$-$1$ matrix. We are promised that either $\theta_N(G)\leq c$ (yes instance) or else $\theta_N(G)\geq c+\epsilon$ (no instance) and we are asked to decide which is the case.
\end{problem}
\end{mdframed}

\begin{theorem}[\cite{BHQMA}] 
\label{thm:reduction}
XY Hamiltonian with local magnetic fields is QMA-complete. Moreover, there is a direct reduction that maps an instance of $3$-Frustration-Free Bose-Hubbard Hamiltonian specified by $G$, $N$, and $\epsilon$, to an instance of XY Hamiltonian specified by $G$, $N$, $\epsilon'=\frac{\epsilon}{4}$, and $c=N\mu(G)+\frac{\epsilon}{4}$, with the same solution.
\end{theorem}

This reduction is presented in Appendix B of \cite{BHQMA}. For our purposes it is crucial that the reduction does not change the underlying interaction graph.  In this paper we show that the $\alpha$-Frustration-Free Bose-Hubbard Hamiltonian on \emph{simple graphs} is QMA-complete, and then we use the above reduction to show that the XY model (on simple graphs, i.e., without local magnetic fields) is QMA-complete.


\section{Adding Self-Loops}
\label{sec:adding_self_loops}

In general, a gate graph is not a simple graph since it may have self-loops.
From equations \eq{HS} and \eq{HE}, we see that self-loops
in the gate graph $G$ arise from both self-loops and edges in its
gate diagram. An edge or a self-loop in the gate diagram is associated
with $16$ self-loops or $8$ self-loops in $G$, respectively. In
this Section we describe a mapping from any $e_{1}$-gate graph $G$
to a modified graph $G^{\mathrm{SL}}$. The graph $G^{\mathrm{SL}}$ is
not a gate graph. It is designed so that it has a self-loop on each of its vertices. We prove
that certain properties of $G^{\mathrm{SL}}$ are related to those of
$G$. Ultimately our goal is to establish a relationship between the ground energies of the Bose-Hubbard models on these graphs.

\subsection{Definition of $G^{\mathrm{SL}}$}

Consider an $e_{1}$-gate graph $G$ described by a gate diagram with $R$ diagram elements and edge and self-loop sets $\mathcal{E}$ and $\mathcal{S}$, respectively.

Define $\mathcal{N}$ to be the set of vertices without self-loops in $G$, i.e.,
\begin{equation}
  \mathcal{N}  = \{ (q,z,t,j) : j\in \{0,\ldots,7\},\, (q,z,t) \notin \mathcal{S} \text{, and } \forall (q',z',t'),\, \{(q,z,t),(q',z',t')\}\notin \mathcal{E}\}.
\label{eq:Ndef}
\end{equation}
Note that $\mathcal{N}$ contains $(q,z,t,j)$ if and only if it contains $(q,z,t,i)$ for all $i\in[8]$, and further that for each $q$, there exists some $t^\star$ such that $(q,z,t^\star,j) \in \mathcal{N}$ for all $z\in \{0,1\}$ and $j\in\{0,\ldots,7\}$.

The vertex set of $G^{\rm SL}$ is two copies of the vertex set of $G$. We label the vertices of $G^{\rm SL}$ as
\[
V(G^{\rm SL}) =
\{(q,z,t,j,d) : q\in[R],\, t\in[8],\, j\in \{0,\ldots,7\},\, z,d \in\{0,1\}\}.
\]
We define $G^{\rm SL}$ by its adjacency matrix 
\begin{equation}
  A(G^{\rm SL}) = A(G) \otimes 1_d+ 2\,\Pi_{\mathcal{N}}\otimes \Pi_+
	\label{eq:adjacency_SL}
	\end{equation}
where
\[
\Pi_{\mathcal{N}} = \sum_{v\in \mathcal{N}} \den{v} \qquad \Pi_+ = \sden{+}=\frac{1}{2}(\ket{0}+\ket{1})(\bra{0}+\bra{1}).
\]
The first term of \eq{adjacency_SL} is just two copies of $G$; the second term adds edges between the two copies as well as self-loops to both copies. Note that every vertex of $G^{\rm SL}$ has a self-loop. 

\subsection{Relationship between the adjacency matrices of $G$ and $G^{\rm SL}$}

Since $A(G^{\rm SL})$ commutes with $1_q\otimes1_z\otimes 1_t\otimes 1_j \otimes \Pi_+$, there is an eigenbasis for $A(G^{\rm SL})$  where each vector is of the form $\ket{\phi}|+\rangle$ or $\ket{\phi}|-\rangle$. For states of the latter form,  the second term in \eq{adjacency_SL} vanishes. From this we see that, if $\sket{\psi}$ is in the ground space of $A(G)$, then $\ket{\psi}\ket{-}$ is in the ground space of $A(G^{\rm SL})$ and its ground energy is $e_1$. We now prove that these states are a basis for the ground space.

\begin{lemma}\label{lem:GSL_ground space} Let $G$ be an $e_1$-gate graph and let $\mathcal{F}$ be the ground space of $A(G)$. Let
\begin{equation}
\mathcal{F}_{\pm}=\{|\alpha\rangle|\pm\rangle: |\alpha\rangle\in \mathcal{F}\}.\label{eq:fpm}
\end{equation}
Then the ground space of $A(G^{\rm SL})$ is $\mathcal{F}_-$, and furthermore, 
\begin{equation}
(2\,\Pi_\mathcal{N}\otimes \Pi_+)\big|_{\mathcal{F}_+} \geq \frac{1}{4}. \label{eq:Fplus}
\end{equation}
\end{lemma}

\begin{proof}
To prove \eq{fpm} it suffices to show that no state of the form $\ket{\beta}\ket{+}$ is in the ground space of $A(G^{\rm SL})$ (since $A(G^{\rm SL})$ commutes with $1_q\otimes1_z\otimes 1_t\otimes 1_j \otimes \Pi_+$).  

Suppose (to reach a contradiction) that $\ket{\beta}\ket{+}$ is in the ground space. Since the ground energy of $A(G^{\rm SL})$ is equal to $e_1$, we must have
\[
\bra{\beta}\bra{+} A(G^{\rm SL}) \ket{\beta}\ket{+}=\bra{\beta}A(G) \ket{\beta}=e_1
\]
and hence 
\begin{equation}
\bra{\beta}\Pi_\mathcal{N} \ket{\beta}=0.
\label{eq:PN}
\end{equation}
Now, since $\ket{\beta}$ is in the $e_1$-energy ground space of $A(G)$, it is contained in the space $\mathcal{Y}$ defined in equation \eq{Yspace}. We now show that 
\begin{equation}
\Pi_\mathcal{N} \big|_{\mathcal{Y}}\geq \frac{1}{8},
\label{eq:pos}
\end{equation}
which contradicts \eq{PN}; this will show that no such state $\ket{\beta}\ket{+}$ exists. It also establishes \eq{Fplus}.

To establish \eq{pos}, first observe that for $q\neq s$ we have $\langle\psi^{s}_{x,b}|\Pi_\mathcal{N}|\psi^{q}_{z,a}\rangle=0$, so the operator $\Pi_\mathcal{N} |_{\mathcal{Y}}$ is block diagonal with a block for each $q\in [R]$ (each block is a $4\times 4 $ principal submatrix corresponding to the subspace spanned by the states $|\psi^q_{z,a}\rangle$ with $z,a\in \{0,1\}$). It is therefore sufficient to establish \eq{pos} for each block individually. 

Focus on the block labeled by some $q_0 \in [R]$. Now we use the fact about $\mathcal{N}$ that is noted in the text following equation \eq{Ndef}, namely, that there exists some $t^\star$  such that $(q_0,z,t^\star,j) \in \mathcal{N}$ for all $z\in \{0,1\}$ and $j\in\{0,\ldots,7\}$.  Using this fact we can write 
\[
\Pi_\mathcal{N}=|q_0\rangle\langle q_0|\otimes 1_z\otimes |t^\star\rangle\langle t^\star|\otimes 1_j+B
\]
for some positive semidefinite operator $B$. To finish the proof of equation \eq{pos}, we show that the first term on the right-hand side is strictly positive within the block labeled by $q_0$, which follows from
\[
\langle \psi^{q_0}_{x,b}| \bigl(|q_0\rangle\langle q_0|\otimes 1_z\otimes |t^\star\rangle\langle t^\star|\otimes 1_j\bigr)|\psi^{q_0}_{z,a}\rangle=\langle \psi_{x,b}|1_z\otimes |t^\star\rangle\langle t^\star|\otimes 1_j|\psi_{z,a}\rangle=\frac{1}{8}\delta_{z,x}\delta_{a,b},
\]
where in the last equality we used equations \eq{psi1} and \eq{psi2}. 
\end{proof}

\begin{lemma} Let $G$ be an $e_1$-gate graph. Then
\[
\gamma(A(G^{\rm SL})-e_1)\geq C \, \gamma(A(G)-e_1)
\]
where $C$ is an absolute constant.
\label{lem:gap_lem}
\end{lemma}

\begin{proof}
We use the Nullspace Projection Lemma. Write $A(G^{\rm SL})-e_1=H_A+H_B$ as a sum of two positive semidefinite operators
\[
H_A=A(G)\otimes 1-e_1 \qquad H_B=2\,\Pi_\mathcal{N}\otimes \Pi_+.
\]
Note that $\|H_B\|=2$ since $\Pi_\mathcal{N}\otimes \Pi_+$ is a projector. We will also need a bound on $\gamma(H_B|_S)$ where $S$ is the nullspace of $H_A$. Note that $S=\mathcal{F}_+ + \mathcal{F}_-$ where $\mathcal{F}_\pm$ is defined in equation \eq{fpm}. By \lem{GSL_ground space}, the nullspace of $H_B|_S$ is $\mathcal{F}_-$. Using this fact, we see that $\gamma(H_B|_S)$ is equal to the smallest eigenvalue of $H_B$ within the space $\mathcal{F}_+$, which we bound using equation \eq{Fplus}:
\[
\gamma(H_B|_S)\geq \frac{1}{4}.
\]

Now applying the Nullspace Projection Lemma, we get
\begin{equation}
\gamma(A(G^{\rm SL})-e_1)\geq \frac{\frac{1}{4}\gamma(A(G)\otimes 1-e_1)}{
\gamma(A(G)\otimes 1-e_1)+2}\geq  \frac{\frac{1}{4}\gamma(A(G)\otimes 1-e_1)}{
\|A(G)\|+|e_1|+2}
\end{equation}
To complete the proof, we show that $\|A(G)\|$ is upper bounded by an absolute constant. Looking at the general expression \eq{AG} for the adjacency matrix of a gate graph, we see that 
\[
\|A(G)\|\leq \|A(g_0)\|+\|h_\mathcal{S}\|+\|h_\mathcal{E}\|\leq \|A(g_0)\|+1+2.
\]
This completes the proof: the right-hand side of this expression is an absolute constant since $g_0$ is a fixed 128-vertex graph.
\end{proof}

\subsection{Relationship between the Bose-Hubbard models on $G$ and $G^{\rm SL}$}

We begin by defining a linear map from the Hilbert space $(\mathbb{C}^{|V(G)|})^{\otimes N}$ of $N$ distinguishable particles on $G$ to the corresponding space $(\mathbb{C}^{|V(G^{\rm SL})|})^{\otimes N}$ for $G^{\rm SL}$. This map is defined by its action on basis states as follows:
\[
|i_1 \rangle |i_2 \rangle\ldots |i_N\rangle \mapsto |i_1,-\rangle|i_2,-\rangle \ldots |i_N,-\rangle
\]
where $|v,-\rangle=\frac{1}{\sqrt{2}}(|v,0\rangle-|v,1\rangle)$ is a superposition of the two vertices $(v,0)$ and $(v,1)$ of $G^{\rm SL}$ that are associated with vertex $v\in V(G)$.

For a state $|\phi\rangle\in (\mathbb{C}^{|V(G)|})^{\otimes N}$ we write $|\overline{\phi}\rangle$ for its image under this mapping, i.e., 
\[
|\phi\rangle =\sum_{x\in V(G)^N}\alpha_x |x_1\rangle|x_2\rangle\ldots |x_N\rangle
\qquad\mapsto\qquad
|\overline{\phi}\rangle =\sum_{x\in V(G)^N}\alpha_x |x_1,-\rangle|x_2,-\rangle\ldots |x_N,-\rangle.
\]
Clearly, if $|\phi\rangle$ is normalized then so is $|\overline{\phi}\rangle$. Furthermore, if $|\phi\rangle\in \mathcal {Z}_N(G)$ (i.e., if it is symmetric under permutations of the $N$ registers) then $|\overline{\phi}\rangle \in \mathcal {Z}_N(G^{\rm SL})$.

\begin{lemma}
Let $|\phi\rangle,|\psi\rangle \in (\mathbb{C}^{|V(G)|})^{\otimes N}$. Then
\[
\langle \overline{\phi}|\Biggl(\sum_{w\in V(G^{\rm SL})}\hat{n}_w(\hat{n}_w-1)\Biggr)|\overline{\psi}\rangle=\frac{1}{2}\langle \phi|\Biggl(\sum_{v\in V(G)}\hat{n}_v(\hat{n}_v-1)\Biggr)|\psi\rangle.
\]
\label{lem:int}
\end{lemma}\begin{proof}
 Writing $(v,0)$ and $(v,1)$ for the two vertices of $V(G^{\rm SL})$ corresponding to a vertex $v\in V(G)$, we have
  \[
    \sum_{w\in V(G^{\rm SL})}\hat{n}_w(\hat{n}_w-1)= \sum_{v\in V(G), d\in \{0,1\}} \hat{n}_{v,d} (\hat{n}_{v,d} - 1) \qquad \hat{n}_{v,d} = \sum_{i=1}^N \den{v,d}^{(i)}.
  \]
  We can then rewrite each term in the sum as
  \begin{align*}
    \hat{n}_{v,d} (\hat{n}_{v,d}  - 1)
	&= \sum_{j\neq j'} \Bigl(\den{v,d}^{(j)}\Bigr)\Bigl(\den{v,d}^{(j')}\Bigr).
  \end{align*}
  
  Using this expression (and a similar expression for $\hat{n}_v(\hat{n}_v-1)$), we see that
  \begin{align*}
    \sbra{\overline{\phi}}  \hat{n}_{v,d} (\hat{n}_{v,d}  - 1)  \sket{\overline{\psi}} &=  \sbra{\overline{\phi}}   \sum_{j\neq j'} \Bigl(\den{v,d}^{(j)}\Bigr)\Bigl(\den{v,d}^{(j')}\Bigr)  \sket{\overline{\psi}}\\
           &= \big(\langle -|d\rangle\langle d|-\rangle\big)^2 \sbra{\phi}   \sum_{j\neq j'} \Bigl(\den{v}^{(j)}\Bigr)\Bigl(\den{v}^{(j')}\Bigr)  \sket{\psi}\\
	&= \frac{1}{4} \sbra{\phi} \hat{n}_v (\hat{n}_v - 1) \sket{\psi}.
  \end{align*}  
  Summing both sides over $d\in\{0,1\}$ and $v\in V(G)$ gives the claimed result.
\end{proof}
\begin{lemma}
Let $G$ be an $e_1$-gate graph. Then
\begin{equation}
\lambda_N^{1}(G)\leq a \quad \implies \quad \lambda_N^{1}(G^{\rm SL})\leq \frac{3}{2}a.
\label{eq:SL_c}
\end{equation}
If in addition $G$ is described by a gate diagram with $R$ diagram elements and satisfies $\gamma(A(G)-e_1)\geq \frac{C_0}{R^3}$ for some absolute constant $C_0$, then 
\begin{equation}
\lambda_N^{1}(G)\geq b>0 \quad \implies \quad \lambda_N^{1}(G^{\rm SL})\geq \frac{C_0C_1}{R^3N^2}b
\label{eq:SL_s}
\end{equation}
where $C_1$ is an absolute constant.
\label{lem:modified_bounds}
\end{lemma}
\begin{proof}
Let $|\phi\rangle\in \mathcal{Z}_N(G)$ be a normalized state with minimal energy for $H(G,N)$, i.e., $\langle \phi|H(G,N)|\phi\rangle=\lambda_N^{1}(G)$. The normalized state $|\overline{\phi}\rangle\in \mathcal{Z}_N(G^{\rm SL})$ satisfies
\[
A(G^{\rm SL})^{(i)}|\overline{\phi}\rangle=(A(G)\otimes1_d+2\,\Pi_{\mathcal{N}}\otimes \Pi_+)^{(i)}|\overline{\phi}\rangle=(A(G)\otimes1_d)^{(i)}|\overline{\phi}\rangle
\]
for each $i\in [N]$, so 
\begin{align*}
\langle \overline{\phi}|\sum_{i=1}^N (A(G^{\rm SL})-e_1)^{(i)}|\overline{\phi}\rangle&=\langle \overline{\phi}|\sum_{i=1}^N (A(G)\otimes1_d-e_1)^{(i)}|\overline{\phi}\rangle\\
&=\langle \phi|\sum_{i=1}^N (A(G)-e_1)^{(i)}|\phi\rangle\\
&\leq \langle \phi|\Biggl(\sum_{i=1}^N (A(G)-e_1)^{(i)}+\sum_{v\in V(G)} \hat{n}_v (\hat{n}_v-1)\Biggr)|\phi\rangle\\
&=a
\end{align*}
where in going from the second to the third line we used the fact that $\hat{n}_v (\hat{n}_v-1)$ is positive semidefinite.  Using this inequality and \lem{int}, we have
\begin{align*}
\lambda_N^{1}(G^{\rm SL})&\leq \langle \overline{\phi}|H(G^{\rm SL},N)|\overline{\phi}\rangle\\
&\leq a+\langle \overline{\phi}|\sum_{w\in V(G^{\rm SL})}\hat{n}_w (\hat{n}_w-1)|\overline{\phi}\rangle\\
&=a+\frac{1}{2}\langle \phi|\sum_{v\in V(G)} \hat{n}_v (\hat{n}_v-1)|\phi\rangle\\
&\leq a+\frac{1}{2}\langle \phi|H(G,N)|\phi\rangle\\
&\leq \frac{3}{2}a,
\end{align*}
which establishes equation \eq{SL_c}. 

Now suppose $\lambda_N^{1}(G)\geq b>0$ and $\gamma(A(G)-e_1)\geq \frac{C_0}{R^3}$ and consider the second part of the Lemma. Note that in this case $H(G,N)$ has no nullspace, so $\lambda_N^{1}(G)=\gamma(H(G,N))$. We write $H(G^{\rm SL},N)=H_A+H_B$ with positive semidefinite operators
\[
H_A=\sum_{i=1}^N (A(G^{\rm SL})-e_1)^{(i)}\Big|_{\mathcal{Z}_N(G^{\rm SL})} \qquad H_B=\sum_{w\in V(G^{\rm SL})}\hat{n}_w (\hat{n}_w-1)\big|_{\mathcal{Z}_N(G^{\rm SL})} 
\]
and we use the Nullspace Projection Lemma.

To get a bound on $\gamma(H_A)$, first note that every eigenvalue of $H_A$ is also an eigenvalue of the operator
\begin{equation}
\sum_{i=1}^N (A(G^{\rm SL})-e_1)^{(i)}
\label{eq:fullspace}
\end{equation}
(without the restriction to $\mathcal{Z}_N(G^{\rm SL})$), since this operator is permutation symmetric and preserves the symmetric subspace. Hence the smallest nonzero eigenvalue of $H_A$ is at least that of \eq{fullspace} and 
\begin{equation}
\gamma(H_A)
\geq \gamma\Biggl(\sum_{i=1}^N (A(G^{\rm SL})-e_1)^{(i)}\Biggr)
= \gamma(A(G^{\rm SL})-e_1)
\geq C\gamma(A(G)-e_1)\geq \frac{C C_0}{R^3}
\label{eq:bnd1}
\end{equation}
where $C$ is an absolute constant (in the second step we used the fact that $\gamma(M\otimes I +I\otimes M)=\gamma(M)$ for any   Hermitian matrix $M$ with smallest eigenvalue zero, and in the third step we used \lem{gap_lem}).

We also need a bound on $\gamma(H_B|_S)$, where $S$ is the nullspace of $H_A$. Letting $T$ be the nullspace of 
\[
\sum_{i=1}^N (A(G)-e_1)^{(i)}\Big|_{\mathcal{Z}_N(G)}
\]
and using \lem{GSL_ground space}, we see that $S$ is equal to the image of $T$ under the mapping $|\phi\rangle\mapsto|\overline{\phi}\rangle$. Using this fact and \lem{int}, we get
\[
H_B|_S=\frac{1}{2}\sum_{v\in V(G)} \hat{n}_v (\hat{n}_v-1)\big|_T.
\]
Since $H(G,N)$ has no nullspace, neither does the operator on the right-hand side of this equation. Hence $H_B|_S>0$ and $\lambda_N^{1}(G^{\rm SL})=\gamma(H(G^{\rm SL},N))$.  Furthermore
\begin{equation}
\gamma(H_B|_S)=\frac{1}{2}\gamma\Biggl(\sum_{v\in V(G)} \hat{n}_v (\hat{n}_v-1)\big|_T\Biggr)\geq\frac{1}{2} \gamma (H(G,N))\geq \frac{b}{2}
\label{eq:bnd2}
\end{equation}
where in the first inequality we used \fct{variation_gap}. 

Now using the bounds \eq{bnd1} and \eq{bnd2} along with the fact that $\|H_B\|\leq N^2$ and applying the Nullspace Projection Lemma, we get
\[
\lambda_N^{1}(G^{\rm SL})=\gamma(H(G^{\rm SL},N))\geq \frac{\frac{C C_0 b}{2R^3}}{\frac{C C_0}{R^3}
+N^2}\geq  \frac{\frac{C C_0 b}{2R^3}}{(C C_0+1)N^2}.
\]
Thus the result follows with $C_1=\frac{C}{2CC_0+2}$.
\end{proof}

\section{Removing self-loops}
\label{sec:removing_self_loops}

Our goal is to consider simple graphs, but so far we have described a method for mapping an $e_1$-gate graph $G$ to a graph $G^{\text{SL}}$ with self-loops on every vertex. We now remove all the self loops from $G^{\text{SL}}$ to obtain a simple graph $G^{\text{NSL}}$. The adjacency matrix of this graph is 
\[
  A(G^{\text{NSL}}) = A(G^{\text{SL}}) - \II
\]
and it has smallest eigenvalue $\mu(G^{\text{NSL}})=\mu(G^{\text{SL}})-1$.

Now consider the $N$-particle Bose-Hubbard model on $G^{\text{NSL}}$. We have
\[
  H_{G^{\text{NSL}}}^N = \sum_{i=1}^N A(G^{\text{NSL}})^{(i)} + \sum_{k\in V(G^{\text{NSL}})} \hat{n}_k (\hat{n}_k - 1)= H_{G^{\text{SL}}}^N - N,
\]
so
\[
  H(G^{\text{NSL}},N) = \bar{H}_{G^{\text{NSL}}}^N - N \mu(G^{\text{NSL}}) = \bigl(\bar{H}_{G^{\text{SL}}}^N - N\bigr) -N\bigl(\mu(G^{\text{SL}})-1\bigr) = H(G^{\text{SL}},N).
\]
In particular, the smallest eigenvalues of these two Hamiltonians are equal:
\begin{equation}
  \lambda_N^1(G^{\text{NSL}}) = \lambda_N^1(G^{\text{SL}}).
\label{eq:equalNSL}
\end{equation}
We now use this relationship to show that the following problem is QMA-complete.  
\begin{mdframed}
\begin{problem}
[\textbf{$\alpha$-Frustration-Free Bose-Hubbard Hamiltonian on simple graphs}]
We are given a $K$-vertex simple graph $G$, a number of particles $N\leq K$, and a precision parameter $\epsilon=\frac{1}{T}$. The integer $T \ge 4K$ is provided in unary. We are promised that either $\lambda_{N}^{1}(G)\leq\epsilon^\alpha$ (yes instance) or $\lambda_{N}^{1}(G)\geq \epsilon+\epsilon^\alpha$ (no instance) and we are asked to decide which is the case. 
\end{problem}
\end{mdframed}

\begin{theorem}
\label{thm:mainthm}
For any positive integer $\alpha$, the problem $\alpha$-Frustration-Free Bose-Hubbard Hamiltonian on simple graphs is QMA-complete.
\end{theorem}

\begin{proof}
Let $\alpha$ be a fixed positive integer. The problem is clearly contained in QMA since it is a special case of the QMA-complete problem $\alpha$-Frustration-Free Bose-Hubbard Hamiltonian. To show that it is QMA-hard, we reduce from another QMA-hard problem. Let $\beta=8\alpha$ and define ``Problem A'' to be the special case of $\beta$-Frustration-Free Bose-Hubbard Hamiltonian where the graph $G$ is promised to be an $e_1$-gate graph satisfying \eq{condition}. \cor{specialcase} implies that Problem A is QMA-hard. We provide a reduction from Problem A to $\alpha$-Frustration-Free Bose-Hubbard Hamiltonian on simple graphs. 

Let an instance of Problem A be given, specified by $G$, $N$, and $\epsilon$. We assume that $\epsilon$ is smaller than some absolute constant $C$. We show that any such instance of Problem A has the same solution as the instance of $\alpha$-Frustration-Free Bose-Hubbard Hamiltonian on the simple graph $G^{\text{NSL}}$, with number of particles $N$ and precision parameter $\epsilon'=\epsilon^7$. This is sufficient to prove QMA-hardness of $\alpha$-Frustration-Free Bose-Hubbard Hamiltonian on simple graphs since there are only finitely many instances of Problem A (and of $\beta$-Frustration-Free Bose-Hubbard Hamiltonian) with $\epsilon$ lower bounded by a constant.

First we check that $G^{\text{NSL}}$, $N$, and $\epsilon'$ satisfy the conditions in the definition of the problem, i.e., that they specify a valid instance of $\alpha$-Frustration-Free Bose-Hubbard Hamiltonian. Let $K=|V(G)|$ and $K'=|V(G^{\text{NSL}})|$. Then $K'=2K$, and since $\epsilon=\frac{1}{T}\leq \frac{1}{4K}$, we have
\[
\epsilon'=\epsilon^7=\frac{1}{T'}
\]
with $T'=T^7 \geq 8K=4K'$. Furthermore, $N\leq K\leq K'$. Therefore $G^{\text{NSL}}$, $N$, and $\epsilon'$ specify a valid instance.

Next we show that, provided $\epsilon$ is smaller than some constant $C$, the instance of $\alpha$-Frustration-Free Bose-Hubbard Hamiltonian defined by $G^{\text{NSL}}$, $N$, and $\epsilon'$  has the same solution as the instance of Problem A defined by $G$, $N$, and $\epsilon$. 

Suppose first that the instance of Problem A is a yes instance. Then $\lambda_N^{1}(G)\leq \epsilon^\beta$. Applying the first part of \lem{modified_bounds} and using \eq{equalNSL}, we get
\[
\lambda_N^{1}(G^{\text{NSL}})\leq \frac{3}{2}\epsilon^\beta=\frac{3}{2}\epsilon^{8\alpha}\leq \epsilon^{7\alpha} = (\epsilon')^{\alpha}
\]
where in the second-to-last inequality we used the fact that $\frac{3}{2}\epsilon^{\alpha}\leq \frac{3}{2}\epsilon\leq 1$, which follows from the fact that $\epsilon^{-1}$ is an integer greater than 1. So in this case $G^{\text{NSL}}$, $N$, and $\epsilon'$ specify a yes instance of $\alpha$-Frustration-Free Bose-Hubbard Hamiltonian.

Next suppose that the instance of Problem A is a no instance, so $\lambda_N^{1}(G)\geq \epsilon+\epsilon^\beta$. Applying the second part of \lem{modified_bounds} and using \eq{equalNSL}, we get 
\[
\lambda_N^{1}(G^{\text{NSL}})\geq \frac{D}{R^3 N^2} (\epsilon+\epsilon^\beta)
\]
where $D$ is an absolute constant and $R$ is the number of diagram elements in $G$. Noting that $R\leq K$, $N\leq K$,  and $\epsilon\leq \frac{1}{4K}$, we see that 
\[
\lambda_N^{1}(G^{\text{NSL}})\geq D'(\epsilon^6+\epsilon^{5+\beta})
\]
where $D'$ is another absolute constant. Now, provided $2\epsilon\leq D'$ (which holds for sufficiently large problem size $K$), we have
\[
\lambda_N^{1}(G^{\text{NSL}})\geq 2\epsilon^7\geq \epsilon^{7}+(\epsilon^{7})^\alpha=\epsilon'+(\epsilon')^\alpha,
\]
which shows that $G^{\text{NSL}}$, $N$, and $\epsilon'$ specify a no instance of $\alpha$-Frustration-Free Bose-Hubbard Hamiltonian.
\end{proof}

Finally, we apply this result to the XY Hamiltonian problem.

\begin{mdframed}
\begin{problem}
[\textbf{XY Hamiltonian on simple graphs}]
We are given a $K$-vertex simple graph $G$, an integer $N\leq K$, a real number $c$, and a precision parameter $\epsilon=\frac{1}{T}$. The positive integer $T$ is provided in unary. We are promised that either $\theta_N(G)\leq c$ (yes instance) or else $\theta_N(G)\geq c+\epsilon$ (no instance) and we are asked to decide which is the case.
\end{problem}
\end{mdframed}

QMA-hardness of this problem now follows directly from the fact that $3$-Frustration-Free Bose-Hubbard Hamiltonian on simple graphs is QMA-hard (from \thm{mainthm}) and the reduction from \thm{reduction}.

\begin{theorem}
	XY Hamiltonian on simple graphs is QMA-complete.
\end{theorem}

\section*{Acknowledgments}

This work was supported in part by CIFAR; NSERC; the Ontario Ministry of Research and Innovation; the Ontario Ministry of Training, Colleges, and Universities; and the US ARO. DG acknowledges funding provided by the Institute for Quantum Information and Matter, an NSF Physics Frontiers Center (NFS Grant PHY-1125565) with support of the Gordon and Betty Moore Foundation (GBMF-12500028).

\bibliographystyle{myhamsplain}
\bibliography{XYQMA}
\appendix

\section{Proof of \cor{gap}} \label{app:cor}

The strategy used in \cite{BHQMA} to establish QMA-hardness is to reduce quantum circuit satisfiability to Frustration-Free Bose-Hubbard Hamiltonian using two steps. First, starting from a QMA verification circuit $\mathcal{C}_X$ for an instance $X$ of a problem in QMA, it was shown how to construct an $e_1$-gate graph $G_X$. (The explicit construction of $G_X$ is detailed in Section 6.2 of \cite{BHQMA}, and the fact that it is an $e_1$-gate graph is proven in Lemma 11 of that paper.) Second, the graph $G_X$ (along with some extra information encoding ``occupancy constraints'') is used to construct another $e_1$-gate graph $G^{\square}_X$ (this construction is detailed in Appendix C of \cite{BHQMA}). The class of instances of Frustration-Free Bose-Hubbard Hamiltonian that are shown to be QMA-hard all use graphs $G^{\square}_X$ of this form. Thus, to prove \cor{specialcase}, it suffices to show that 
\begin{equation}
\gamma\bigl(A(G_X^{\square})-e_1\bigr)\geq \frac{C_0}{R^3}
\label{eq:boundagx}
\end{equation}
where $R$ is the number of diagram elements in $G^{\square}_X$ and $C_0$ is an absolute constant.

To prove the following results, we rely on some facts established in \cite{BHQMA} (in particular, see Sections 5, 6.2, 7.1, and Appendix E.4).
In this appendix we analyze the spectral gaps of $G_X$ and $G^\square_X$ and thereby establish \eq{boundagx}.

\begin{lemma}
$\gamma(A(G_X)-e_1)\geq D$, where $D$ is an absolute constant.
\label{lem:D_lem}
\end{lemma}
\begin{proof}
Section 6.2 of \cite{BHQMA} describes in detail how the graph $G_X$ is constructed from the verification circuit. Its adjacency matrix can be written
\[
A(G_X)=A(G_1)+h_{\mathcal{E^\prime}}+h_{\mathcal{S^\prime}}
\]
where $G_1$ is an $e_1$-gate graph with many components, each of constant size. Here $h_{\mathcal{E^\prime}}$ and $h_{\mathcal{S^\prime}}$ are of the form given in equations \eq{HS} and \eq{HE}, but in this case, $\mathcal{E^\prime}$ and $\mathcal{S^\prime}$ are subsets all the edges and self-loops, respectively, in the gate diagram for $G_X$. In particular, the gate diagram for $G_1$ contains some of the edges and self-loops in the gate diagram for $G_X$, and the rest are included in the sets  $\mathcal{E^\prime}$ and $\mathcal{S^\prime}$.

Note that in reference \cite{BHQMA} the terms we refer to as $h_{\mathcal{E^\prime}}$ and $h_{\mathcal{S^\prime}}$ above are denoted by
\begin{equation}
h_{\mathcal{E^\prime}}+h_{\mathcal{S^\prime}}=h_1+h_2+\sum_{i=n_{\rm{in}}+1}^{n}h_{\rm{in},i}+h_{\rm out}.
\label{eq:decomp}
\end{equation}
Each of the terms on the right-hand side corresponds to a subset of the edges or self-loops in the gate diagram for $G_X$; in reference  \cite{BHQMA}  it was convenient to further subdivide the sets $\mathcal{E^\prime}$ and $\mathcal{S^\prime}$, but it will not be necessary to review the details of this here.

Every component of $A(G_1)$ is one of 5 fixed graphs, each with a constant number of vertices (four of these 5 possible graphs are called two-qubit gate gadgets and the fifth is called the boundary gadget--see Figures 5.2-5.3 of \cite{BHQMA}). The smallest nonzero eigenvalue satisfies 
\begin{equation}
\gamma(A(G_1)-e_1)\geq c,
\label{eq:c_bnd}
\end{equation}
where $c$ is an absolute constant equal to the smallest nonzero eigenvalue of one of these 5 graphs.  

A basis for the $e_1$-energy ground space of $A(G_1)$, which we denote by $S$ in the following, is given by the set of states $\{|\rho_{z,a}^L\rangle : L\in \mathcal{L}\}$ defined before Lemma 11 in \cite{BHQMA}. In this basis for $S$ one can compute the matrix elements of the operator
\begin{equation}
(h_{\mathcal{E^\prime}}+h_{\mathcal{S^\prime}})\big|_S
\label{eq:rest_S}
\end{equation}
These matrix elements are computed in reference \cite{BHQMA}. Specifically (referring to the decomposition on the right-hand side of equation \eq{decomp}), the matrix elements of $h_1|_S$ are given in equations (E29)--(E34), those of $h_2|_S$ are given in equations (E39)--(E40), those of $h_{\rm{in},i}|_S$ appear in the equation preceding before (E45), and those of $h_{\rm out}|_S$ are given in the second equation of Section E.4.4.

Here we use the following fact about the operator \eq{rest_S}, which follows from the equations of reference \cite{BHQMA} listed above: it is block diagonal in the basis $\{|\rho_{z,a}^L\rangle : L\in \mathcal{L}\}$ with blocks of size at most two. Furthermore, each of these blocks is one of a fixed set of $2\times 2$ (or $1\times 1$) matrices, and therefore the smallest nonzero eigenvalue of \eq{rest_S}, equal to the smallest nonzero eigenvalue of one of these blocks, is lower bounded by a fixed constant that we denote $\tilde{c}$. Now we apply the Nullspace Projection Lemma using this fact, equation \eq{c_bnd}, and the fact that 
\[
\|h_\mathcal{E'}+h_\mathcal{S^\prime}\|\leq \|h_\mathcal{E'}\|+\|h_\mathcal{S^\prime}\|\leq 2+1.
\]
We get
\[
\gamma(A(G_X)-e_1)\geq \frac{c\tilde{c}}{c+
3},
\]
which completes the proof.
\end{proof}

The following lemma uses a construction from Appendix C.3 of \cite{BHQMA}.
\begin{lemma}
$\gamma(A(G_X^{\square})-e_1)\geq \frac{C_0}{R^3}$, where $R$ is the number of diagram elements in the gate diagram for $G_X^{\square}$ and $C_0$ is an absolute constant.
\end{lemma}
\begin{proof}
The adjacency matrix of the graph $G_X^{\square}$ can be written
\[
A(G_X^{\square})=A(G_X^{\triangle})+h_{\mathcal{E}_0}+h_{\mathcal{S}_0},
\]
where $G_X^{\triangle}$ is an $e_1$-gate graph (it is described in equation (C.8) of \cite{BHQMA} and is used to enforce the occupancy constraints, although we will not need its explicit form) and $h_{\mathcal{E}_0}$ and $h_{\mathcal{S}_0}$ are of the form in equations \eq{HE} and \eq{HS}, respectively, for some set of edges and self-loops in the gate diagram for $G_X^{\square}$. States with low energy for the graph $G_X^{\triangle}$ are related to those with low energy for $G_X$ in the following way.

Since $G_X$ is an $e_1$-gate graph, its adjacency matrix can also be written as
\[
A(G_X)=1_q\otimes A(g_0)+h_{\mathcal{E}}+h_{\mathcal{S}}.
\]
Let $Q$ be the $e_1$-energy ground space of the first term in the above equation (i.e., $1_q\otimes A(g_0)$), and let $Q^\triangle$ be the $e_1$-energy ground space of $A(G^{\triangle})$. The graph $G_X^{\triangle}$ is constructed so that these two spaces have the same dimension, and furthermore
\[
(h_{\mathcal{E}_0}+h_{\mathcal{S}_0})\big|_{Q^{\triangle}}=(h_{\mathcal{E}}+h_{\mathcal{S}})\big|_{Q} \cdot \begin{cases} \frac{1}{3R'+2} & R'\text{ odd}\\
\frac{1}{3R'-1} & R' \text{ even}
\end{cases}
\]
where $R'$ is the number of diagram elements in the gate diagram for $G_X$ (this is equation C.51 of \cite{BHQMA}). It is also proven in Lemma 19 of \cite{BHQMA} that
\begin{equation}
\gamma\bigl(A(G_X^{\triangle})-e_1\bigr)\geq \frac{1}{(30R')^2}.
\label{eq:Gtrigap}
\end{equation}

Using this equation, we get
\begin{equation}
\gamma \Bigl((h_{\mathcal{E}_0}+h_{\mathcal{S}_0})\big|_{Q^{\triangle}}\Bigr)
\geq \frac{1}{4R'} \gamma\Bigl((h_{\mathcal{E}}+h_{\mathcal{S}})\big|_Q \bigr)
\geq \frac{1}{4R'}\gamma(A(G_X)-e_1)\geq \frac{D}{4R'}
\label{eq:gam}
\end{equation}
where in the second inequality we used \fct{variation_gap} and where $D$ is the absolute constant from \lem{D_lem}. 

We now apply the Nullspace Projection Lemma with the decomposition $A(G_X^{\square})-e_1=H_A+H_B$, where $H_A=A(G_X^{\triangle})-e_1$ and $H_B=h_{\mathcal{E}_0}+h_{\mathcal{S}_0}$. Note that $\|H_B\|\leq \|h_{\mathcal{E}}\|+\|h_{\mathcal{S}}\|\leq 3$ . Applying the Nullspace Projection Lemma and using \eq{gam} and \eq{Gtrigap}, we get
\[
\gamma\bigl(A(G_X^{\square})-e_1\bigr)\geq \frac{\frac{D}{4R'(30R')^2}}{
\frac{1}{(30R')^2}+3}\geq\frac{C_0}{R'^3}
\]
where $C_0$ is an absolute constant. To complete the proof, we use the fact that $R$ (the number of diagram elements in $G_X^{\square}$) is greater than or equal to $R'$ (the number of diagram elements in $G_X$), which is a general property of the mapping $G \mapsto G^{\square}$ as described  in Appendix C of \cite{BHQMA}.
\end{proof}

\end{document}